\documentclass[authoryear,a4paper, 12pt]{elsarticle}
\usepackage{natbib}
\usepackage{amsthm}
\usepackage{amssymb}
\usepackage{graphicx,amsmath}
\usepackage{tikz}
\usepackage{mathtools}

\usepackage{booktabs,xcolor}
\usepackage{graphics}
\usepackage{enumitem}
\usepackage{subcaption}
\usepackage{appendix}
\usepackage{adjustbox}
\usepackage{caption}
\usepackage{setspace}
\usepackage[colorlinks=true,linkcolor=blue, urlcolor=blue, citecolor=blue, breaklinks]{hyperref}
\journal{Social Choice and Welfare.}

\DeclareRobustCommand{\da}{\ensuremath{\mathtt{DA}}}
\DeclareRobustCommand{\ia}{\ensuremath{\mathtt{BM}}}
\DeclareRobustCommand{\rsd}{\ensuremath{\mathtt{SD}}}

\newtheorem{proposition}{Proposition}

\newtheorem{conjecture}{Conjecture}

\theoremstyle{definition}

\usepackage{etoolbox}

\patchcmd{\emailauthor}{(#2)}{}{}{}
\patchcmd{\urlauthor}{(#2)}{}{}{}

\newcommand{\rk}{\textup{rk}} 

\newcommand{\ttt}{\mathtt f} 
\newcommand{\rrr}{\mathtt r} 
\newcommand\citepos[1]{\citeauthor{#1}'s (\citeyear{#1})}
\begin{document}
	
	\begin{frontmatter}
		\title{A Note on the Strategic Vulnerability of the Boston Mechanism in Random Markets\tnoteref{}}
		\tnotetext[]{I am grateful to the Associate Editor and two anonymous reviewers for their excellent comments, which significantly improved the paper. Accepted version to \emph{Social Choice and Welfare}.}
		\author{Josué Ortega}
		%\cortext[mycorrespondingauthor]{Corresponding author}
		\ead{j.ortega@qub.ac.uk.}
		\address{Queen's University Belfast, UK.}
		
		\begin{abstract}
		We provide the first asymptotic analysis of the Boston Mechanism under equilibrium play in random markets. We provide two results. 
		
		First, while 63\% of students receive their first preference under truthful reporting—outperforming any other known mechanism in the literature—this rate converges to zero in any Nash equilibrium of the corresponding preference revelation game as the market size grows.

		Second, we show there exists a Nash equilibrium where the average student receives a dramatically inferior assignment: in markets with 1,000 students, the average placement shifts from the 7th choice (under truthfulness) to the 145th choice, representing a change from logarithmic to nearly linear average rank.

		\end{abstract}
		
	\end{frontmatter}
	
	\newpage
	\setcounter{footnote}{0}
	\setlength{\parskip}{1em} 
	
	\section{Introduction}
	\label{sec:introduction}
	School choice mechanisms are central to education policy, with the Boston Mechanism (BM, also known as Immediate Acceptance) remaining one of the most widely adopted approaches despite concerns about its incentive properties \citep{abdulkadirouglu2003}. 
	A central metric for evaluating these mechanisms—and one often emphasized by policymakers—is the \textbf{first preference rate}: the proportion of students assigned to their most preferred school. 
	For example, England's Department for Education publishes the {first preference rate} for primary and secondary schools annually since 2014.\footnote{See  \url{https://explore-education-statistics.service.gov.uk/find-statistics/secondary-and-primary-school-applications-and-offers}.}  A second important efficiency metric is the \textbf{average rank}, which measures the mean position of assigned schools in students' preference rankings. While less commonly reported by practitioners, this metric is widely studied in the academic literature as it captures information about assignment quality across the full distribution of outcomes.\footnote{Although England does not publish the average rank, it does publish the second and third preference rate in the link above.}
	
	In this note, we quantify BM's first preference rate and average rank in every Nash equilibrium of its corresponding preference revelation game under complete information. Our paper complements \citeauthor{pritchard2023asymptotic}'s (\citeyear{pritchard2023asymptotic}, henceforth PW) seminal work analyzing both metrics under the BM, albeit under truthful preference revelation. We find that accounting for strategic behavior significantly changes their results.
	
	We focus on random markets with preferences and priorities drawn independently and uniformly at random and with $n$ students and $n$ schools. In this setup, PW show that BM's first preference rate converges to a constant ($1 - e^{-1} \approx 0.63$) as the market size grows large, assuming truthful preference revelation. Their result demonstrates BM's superior first preference rate compared to other well-known Pareto-efficient mechanisms. For instance, Serial Dictatorship (SD), Top Trading Cycles (TTC) and Rank-Minimizing (RM) all attain a first preference rate of only 0.5 asymptotically \citep{ortega2023cost}.\footnote{
	Other mechanisms are also first-choice maximal, but BM is the only well-known mechanism in this class \citep{dur2018first}.} However, under strategic behavior, this advantage vanishes completely: in \emph{any} Nash equilibrium of the corresponding preference revelation game, BM's first preference rate converges to zero at a rate of $1/\log(n)$, highlighting the extreme vulnerability of the Boston mechanism to strategic manipulation.
	
	Our second finding concerns the expected placement of the average student under BM. Under truthful behavior, PW show that BM's average placement is better (i.e. smaller) than under serial dictatorship, which is in the order of $\log n$. However, we show that when accounting for strategic behavior, the effect of strategic play on average placement depends on the specific equilibrium selection. 
	While there exists a Nash equilibrium of BM's preference revelation game in which the average rank is around $\log n$ too, leading to a small difference with the truthful case, there exist other equilibria where the average rank increases substantially on the order of $n/\log n$, leading to stark differences between the truthful and strategic analysis. 
	To put these findings in context: in a market with $n=1000$ students and schools, the average student is assigned to their 7-th choice under truthfulness, while under some Nash equilibrium they are assigned to their 145-th choice.
	
	Our findings quantify BM's well-known vulnerability to strategic behavior, showing significant effects on performance statistics widely used by education policymakers.
		
	%%%%%%%%%%%%%%%%%%%%%%%%%%%%%%%%%%%%
	\section{Related Literature}
	\label{sec:literature}
	%%%%%%%%%%%%%%%%%%%%%%%%%%%%%%%%%%%%
	
	The random markets approach has proven powerful for analyzing allocation mechanism performance, enabling tractable asymptotic analysis while providing insightful theoretical results. This framework has been particularly fruitful for studying the Deferred Acceptance (DA) mechanism, where interdisciplinary research spanning economics and computer science has built upon the seminal works of \citet{wilson1972} and \citet{knuth1976} to produce an extensive literature \citep{pittel1989,pittel1992likely,immorlica2005,kojima2009incentives,lee2016incentive,liu2016,che2019efficiency,ashlagi2017,pycia2019evaluating,nikzad2022rank,ortega2023cost,ortega2025identifying,ronen2025stable}.
	
	Similarly, the Top Trading Cycles algorithm has been studied within this framework \citep{knuth1996,frieze1995probabilistic,che2017top}, and deep equivalence results have extended some of these findings to arbitrary strategy-proof and Pareto-efficient mechanisms such as Serial Dictatorship \citep{pycia2019evaluating,che2018payoff}. A comparable analysis has been conducted for the Rank-Minimizing mechanism \citep{nikzad2022rank,sethuramannote,ortega2023cost}, as well as for the Rawlsian mechanism \citep{demeulemeester2022rawlsian}.
	
	Despite this extensive asymptotic literature, the Boston Mechanism has remained largely unstudied in large random markets, creating a notable gap given its continued widespread practical use. This gap is particularly striking because BM remains a heavily studied and implemented mechanism in school choice \citep{abdulkadirouglu2003,chen2006school,pathak2008leveling,miralles2009school,afacan2013alternative,kumano2013strategy,kojima2014boston,he2015gaming}. In particular, \citet{ergin2006} showed that the set of Nash equilibrium outcomes in BM's preference revelation game with complete information coincides precisely with the set of stable matchings under true preferences, a result that serves as the basis of our analysis. Yet while the theoretical literature has extensively analyzed BM's strategic properties in finite markets, its asymptotic performance under equilibrium play remained unexplored.
	
The first systematic analysis of BM in random markets appeared in \citet{mennle2017trade}, who establish BM's rank dominance over DA in finite markets whenever they are comparable, occurring in 8\% of preference profiles. They use data from Mexico city school allocation to validate their theoretical findings. \citet{pritchard2023asymptotic} then extended this to asymptotic analysis, deriving exact recurrence formulas to characterize the number of students matched in each round and proving BM's first preference rate converges to $1-e^{-1} \approx 0.63$. They further explore aspects beyond our scope, such as utilitarian welfare under scoring rules. Together, these seminal contributions demonstrate BM's superior welfare performance—but exclusively under truthful behavior. 

Our paper adopts the same random markets framework where preferences and priorities are drawn independently and uniformly at random, but extends the analysis to strategic behavior. This allows us to mirror \citepos{pritchard2023asymptotic} asymptotic results under equilibrium play and quantify how strategic incentives reshape BM's documented welfare advantages, bridging the gap between BM's known manipulability and its performance in large random markets.
	
	%%%%%%%%%%%%%%%%%%%%%%%%%%%%%%%%%%%
	\section{Model }
	\label{sec:model}
	%%%%%%%%%%%%%%%%%%%%%%%%%%%%%%%%%%%
	
	For simplicity, we focus on a one-to-one matching problem, which consists of:
	\begin{enumerate}
		\item A set of students $I = \{i_1, \ldots, i_n\}$,
		\item A set of schools $S = \{s_1, \ldots, s_n\}$, with each school having capacity for one student only,
		\item Strict students' preferences over schools, $\succ \coloneqq (\succ_1, \ldots, \succ_n)$, and
		\item Strict schools' priorities over students, $\triangleright \coloneqq (\triangleright_{s_1}, \ldots, \triangleright_{s_n})$.
	\end{enumerate}
	
	A \emph{matching problem} $P$ is a four-tuple $(I,S,\succ,\triangleright)$. We refer to $n$ as the size of the problem. An \emph{allocation} $x$ assigns each student to exactly one school. We denote by $x_i$ student $i$'s assigned school.

	The function $\rk_i(x_i)$ returns an integer between 1 and $n$ corresponding to the ranking of $x_i$ in student $i$'s preference list (with 1 representing the most preferred school). A \emph{mechanism} is a map from matching problems to allocations.
	
	A \emph{mechanism} $M$ returns an allocation $M(P)$ for every matching problem $P$. We mainly focus on the Boston Mechanism (BM), which is also known as Immediate Acceptance. It operates as follows:
	
	\begin{itemize}
		\item In the first round, each student applies to their most preferred school. Each school that receives at least one application immediately and permanently accepts the applicant with the highest priority, removing both itself and the matched student from the market.
		
		\item In each subsequent round $k$, each remaining student applies to their $k$-th preferred school. Each remaining school that receives at least one application immediately and permanently accepts the highest-priority applicant.
		
		\item The process continues until all students are assigned to a school.
	\end{itemize}
	
	We compare BM to the student-proposing Deferred Acceptance (DA). DA operates similarly to BM, but schools’ acceptances are temporary: a student may be rejected in later rounds if a higher-priority applicant applies.
	We denote by $\ia(P)$ and $\da(P)$ the allocation produced by BM and DA for the matching problem $P$, respectively.
	
	\subsection{Performance Metrics}
	We are interested in the \emph{first preference rate} $\ttt(M(P))$, defined as the fraction of students assigned to their most preferred school in the allocation $M(P)$:
	\begin{equation}
		\ttt (M(P)) \coloneqq \frac{|\{i \in I: \operatorname{rk}_i(x_i)=1\}|}{n}
	\end{equation}
	
	Additionally, we consider the \emph{average rank}, denoted by $\rrr(M(P))$, which is the average position of the assigned school in students' preference lists:
	\begin{equation}
		\rrr(M(P)) \coloneqq \frac{1}{n} \sum_{i \in I} \operatorname{rk}_i(x_i)
	\end{equation}

\subsection{The Preference Revelation Game}
Following \citet{ergin2006}, we model strategic behavior in the Boston mechanism via the induced preference revelation game under complete information. 
Each student $i \in I$ chooses a strict ranking $\succ_i$ over $S$. 
Given $\succ$, the Boston mechanism applied to the matching problem $(I,S,\succ,\triangleright)$ produces the allocation $\ia(I,S,\succ,\triangleright)$.

We adopt a complete information environment: before submitting their lists, all students know the true preference profile $\succ$ and all school priorities $\triangleright$. We maintain this assumption throughout. We discuss the robustness of our results to incomplete information in the Conclusion.

A Nash equilibrium of the preference revelation game induced by the Boston mechanism is a profile $\succ^\ast$ such that for every student $i \in I$ and every alternative ranking $\succ'_i$, the assignment under $\ia(I,S,\succ^\ast,\triangleright)$ is at least as good for $i$, according to the true preferences $\succ_i$, as the alternative assignment under $\ia(I,S,(\succ'_i,\succ^\ast_{-i}),\triangleright)$.

\subsection{Random Markets}
To analyze the expected behavior of a mechanism—rather than its performance on a specific matching problem—we consider a random matching problem in which strict preferences and priorities are drawn independently and uniformly at random. While this assumption abstracts from real-world complexities, it serves as a standard theoretical benchmark that enables tractable analysis and has been extensively used in the literature discussed in Section~\ref{sec:literature}, including PW's work that serves as the starting point of our analysis. 

\paragraph{Truthful Behavior} 
We denote by $\ttt(\ia_n)$ and $\rrr(\ia_n)$ the expected first preference rate and average rank in a random matching problem of size $n$ under BM with truthful preference revelation. The expectation is taken over the random draw of preferences and priorities. For any other mechanism $M$, we define $\ttt(M_n)$ and $\rrr(M_n)$ analogously.

\paragraph{Strategic Behavior}
To define performance under strategic play consistently across all realizations of $(\succ,\triangleright)$, we introduce an \emph{equilibrium selection rule}: a function $\sigma$ that assigns to each preference–priority pair one of its Nash equilibrium outcomes (equivalently, one of its stable matchings). Given such a rule $\sigma$, let $\ttt(\ia_n^\sigma)$ and $\rrr(\ia_n^\sigma)$ denote the expected first preference rate and average rank associated with the equilibrium outcome selected by $\sigma$ at each realized $(\succ,\triangleright)$, where the expectation is again taken over the random draw of preferences and priorities.

Our results characterize performance across all possible equilibrium selection rules. Proposition~\ref{prop:top_choice_strategic} establishes that $\lim_{n \to \infty} \ttt(\ia_n^\sigma) = 0$ for \emph{every} selection rule $\sigma$, so that the vanishing first preference rate is universal across equilibrium selections. In contrast, Proposition~\ref{prop:avg_strategic} shows that the limiting average rank depends on the equilibrium selection: we characterize its range by computing the infimum and supremum over all selection rules $\sigma$.

	%%%%%%%%%%%%%%%%%%%%%%%%%%%%%%%%%%%
	\section{Results}
	\label{sec:truthful}
	%%%%%%%%%%%%%%%%%%%%%%%%%%%%%%%%%%%
	\subsection{First Preference Rate}
	
	\paragraph{Truthful Behavior}
	We begin by describing PW's result quantifying BM's first preference rate under truthful behavior.
	
	\begin{proposition}[\cite{pritchard2023asymptotic}]
		\label{prop:top_choice_truth}
		Under truthful preference revelation, the expected first preference rate of the Boston mechanism approaches $1-e^{-1} \approx 0.63$ as the market size grows large:
		\begin{equation}
			\lim_{n \to \infty} \ttt(\ia_n) = 1 - e^{-1}
		\end{equation}
	\end{proposition}
	
	The intuition is straightforward. In the Boston mechanism, a student is assigned to their top choice if and only if they are assigned in the first round, and the students who are assigned in the first round equals the number of schools that receive at least one application. Since each student applies uniformly at random to one of the $n$ schools, the probability that a school receives no applications from any student is $(1-1/n)^n$, which converges to $e^{-1}$ as $n \to \infty$.
	
	Proposition \ref{prop:top_choice_truth} is particularly interesting because it shows that the Boston mechanism assigns a substantially higher fraction of students to their first preference compared to other standard efficient mechanisms such as serial dictatorship.
	\paragraph{Strategic Behavior}
	\label{sec:strategic}
	While the Boston mechanism exhibits an impressive first preference rate under truthful preference revelation, its well-known manipulability raises questions about its performance when students behave strategically. We analyze how strategic behavior affects the first preference rate in equilibrium.
	
	\begin{proposition}
		\label{prop:top_choice_strategic}
		For every equilibrium selection rule~$\sigma$, the expected first preference rate of the Boston mechanism under equilibrium play converges to zero at least as fast as $\frac{1}{\log n}$ as the market size grows:\footnote{The $O$ notation provides an upper bound on the growth rate of a function. 
			Formally, $f(n) = O(g(n))$ if there exist positive constants $c$ and $n_0$ such that 
			$|f(n)| \leq c \cdot |g(n)|$ for all $n \geq n_0$.}
		
		\begin{equation}
			\ttt(\ia_n^\sigma) = O\!\left(\frac{1}{\log n}\right),
			\qquad
			\text{and hence } \lim_{n \to \infty} \ttt(\ia_n^\sigma) = 0.
		\end{equation}
	\end{proposition}
	
	\begin{proof}
		The proof follows from three observations. 
		First, as shown by \citet{ergin2006}, the set of Nash equilibrium outcomes of the preference revelation game induced by the Boston mechanism coincides exactly with the set of stable matchings under the true preferences. Hence, any equilibrium allocation selected by~$\sigma$ must be stable.
		
		Second, from \citet{gale1962}, the student-proposing Deferred Acceptance (DA) outcome is weakly preferred by every student to any other stable allocation. Therefore, for every selection rule~$\sigma$, we have $\ttt(\ia_n^\sigma) \leq \ttt(\da_n)$.
		
		Third, in random markets, the expected number of applications submitted by a student—and equivalently, the expected number of applications received by each school—is $\Theta(\log n)$ \citep{knuth1976, pittel1989}.\footnote{The $\Theta$ notation provides both upper and lower bounds on the growth rate of a function. Formally, $f(n) = \Theta(g(n))$ if there exist positive constants $c_1, c_2, n_0$ such that $c_1 \cdot g(n) \leq f(n) \leq c_2 \cdot g(n)$ for all $n \geq n_0$.}
		Consequently, the probability that a student is assigned to their top choice under DA is $\Theta(1/\log n)$: each school receives approximately $\log n$ applications, so a student has probability $1/\log n$ of having the highest priority among them.\footnote{This result is {\it ``almost deterministic''} because the number of applications concentrates sharply around its expected value \citep{motwani1995}.} 
		Since $\ttt(\ia_n^\sigma) \leq \ttt(\da_n)$ and $\ttt(\da_n) = \Theta(1/\log n)$, it follows that $\ttt(\ia_n^\sigma) = O(1/\log n)$, and therefore $\lim_{n \to \infty} \ttt(\ia_n^\sigma) = 0$.
	\end{proof}

	This result reveals a stark contrast between truthful and strategic behavior in the Boston mechanism. While the expected first preference rate under truthful reporting converges to approximately 63\% as market size grows, under strategic behavior it converges to zero. This difference highlights how the incentives created by the Boston mechanism transform its properties substantially when participants behave strategically.
	
	\paragraph{Intuition}  
	The collapse in Proposition~\ref{prop:top_choice_strategic} stems from a cascading effect triggered by strategic behavior.  
	Under truthful reporting, students act naively: everyone ``gambles'' on their favorite school in the first round, and since applications are uniformly random, almost every school receives at least one applicant, producing the familiar $1 - e^{-1} \approx 0.63$ first preference rate.  
	
	In equilibrium, however, students anticipate rejection at popular schools and avoid wasting their first-round application where their priority is low. Each such deviation potentially displaces another student from the school they would have obtained under truth-telling, who in turn shifts their application to a lower-ranked school, which itself was potentially someone's else top choice. This chain reaction propagates through the market, eliminating the large set of ``lucky'' first-round acceptances that sustain the 0.63 constant. The resulting configuration is stable: every student applies only where they can actually be accepted given others’ behavior, which coincides with a stable matching.  
	
	Because in large random markets each school receives roughly $\log n$ total applications, only about $1 / \log n$ of students have the highest priority at their true first-choice school, causing the first preference rate to vanish asymptotically. In other words, the Boston mechanism’s remarkable performance under truthful play unravels once agents stop behaving naively and begin responding strategically to their relative priorities.
	
%%%%%%%%%%%%%%%%%%%%%%%%%%%%%%%%%%%
\subsection{Average Rank}
\label{sec:avg_rank}
%%%%%%%%%%%%%%%%%%%%%%%%%%%%%%%%%%%
\paragraph{Truthful Behavior} 
Beyond the first preference rate, the average rank is also a valuable performance statistic that allows us to assess how satisfactory the average student placement is. A well-known result is that Serial Dictatorship achieves an expected average rank of $\Theta(\log n)$ in random markets. PW provide theoretical and simulation evidence suggesting that, under truthful reporting, the Boston mechanism's expected average rank also grows as $\Theta(\log n)$, but with a smaller constant factor—meaning BM yields better (lower) average ranks than Serial Dictatorship.

\begin{conjecture}[\citealt{pritchard2023asymptotic}, p.~269]
	\label{prop:avgrank}
	Under truthful preference revelation in random markets, the expected average rank of the Boston mechanism grows as $\Theta(\log n)$, but with a smaller constant factor than Serial Dictatorship:
	\begin{equation}
		\rrr(\ia_n) = c_{\mathrm{BM}} \cdot \log n,
		\quad
		\rrr(\rsd_n) = c_{\mathrm{SD}} \cdot \log n,
		\quad
		\text{where } c_{\mathrm{BM}} < c_{\mathrm{SD}}.
	\end{equation}
\end{conjecture}

\paragraph{Strategic Behavior}
We now study how strategic behavior affects the average placement. In this setting, the performance of the Boston mechanism depends on the equilibrium selection rule~$\sigma$. We obtain the following result:

\begin{proposition}
	\label{prop:avg_strategic}
	For every market size $n$, let $\rrr(\ia_n^\sigma)$ denote the expected average rank of the Boston mechanism under equilibrium selection rule~$\sigma$. Then:
	\begin{align}
		\inf_{\sigma}\, \rrr(\ia_n^\sigma) &= \Theta(\log n), \\
		\sup_{\sigma}\, \rrr(\ia_n^\sigma) &= \Theta\!\left(\frac{n}{\log n}\right).
	\end{align}
\end{proposition}

\begin{proof}
	The argument parallels the proof of Proposition~\ref{prop:top_choice_strategic}. 
	For any matching problem, both the student-optimal and the school-optimal stable matchings are implementable as Nash equilibria of the Boston mechanism \citep{ergin2006}. 
	The student-optimal stable matching yields an expected average rank of $\Theta(\log n)$, while the school-optimal stable matching yields $\Theta(n / \log n)$ \citep{knuth1976,pittel1989}. 
	Because students weakly prefer the student-optimal stable matching to any other stable matching, and weakly disprefer the school-optimal one, these two allocations achieve the infimum and supremum of expected average ranks across all equilibrium selection rules~$\sigma$.
\end{proof}

Proposition \ref{prop:avg_strategic} reveals that strategic behavior can undermine average rank less severely than the first preference rate. Unfortunately, the worst-case average rank would likely emerge if the number of students is higher than of school seats, as \cite{ashlagi2017} show that in unbalanced markets where students are in excess supply, all stable matchings yield average ranks as poor as the school-optimal one.

\paragraph{Intuition}  
The first preference rate obscures a crucial feature of the Boston mechanism: even under naive behavior, it operates through a built-in cascade.  
Students who fail to secure their favorite school in the first round immediately reapply to their next option, which may already be filled by higher-priority applicants, forcing them to move again in subsequent rounds.  
This displacement process continues until every student finds a remaining vacancy, generating an average rank of order $\Theta(\log n)$ rather than a constant.  
Hence, while the first preference rate under truth-telling appears remarkably high, the average rank already reveals the underlying cascade structure of Boston and highlights its efficiency is not as good as it may seem at first (and far from the first-best constant average rank achieved by the rank-minimizing mechanism).  

In equilibrium, BM's application cascade becomes anticipatory: students foresee where they will be rejected and skip those schools entirely.
If the preemptive adjustments are mild, the cascade stops early and the resulting allocation roughly coincides with the student-optimal stable matching, preserving the $\Theta(\log n)$ average rank.  
If the adjustments propagate deeply, the cascade extends through the entire market, pushing almost every student far down their preference list and leading to the school-optimal stable matching with average rank $\Theta(n / \log n)$.  
Thus, unlike Proposition \ref{prop:top_choice_strategic}, where strategic reasoning creates a cascade that destroys the 0.63 constant, Proposition \ref{prop:avg_strategic} shows that Boston already contains a cascade, and strategic behavior merely determines how far it runs.

\section{Discussion}

\paragraph{Conclusion}
We have quantified how strategic behavior affects the strong performance of the Boston Mechanism in random markets with respect to two key performance statistics. Our main finding shows that, while the Boston Mechanism assigns more students to their first choice than any other known mechanism under truthful reporting, this advantage completely disappears once strategic behavior is taken into account. The effect of strategic play on average placement is more nuanced, as the outcome depends on the specific equilibrium selection.

In establishing these stark results, we have used the standard random market model that serves as the workhorse in the literature and which allows us to compare our results to previous work. While this model assumes independent preferences and priorities, the scale of the welfare loss we document suggests that BM's strategic problems are fundamental rather than artefacts of our modelling choices. Future work could explore whether our results extend to more general environments, such as those with correlated preferences or many-to-one matching, but would require additional tools for their analysis.

To conclude, the remarkable discrepancy between the truthful and strategic settings underscores the importance of accounting for incentives in the study of allocation mechanisms and serves as a first step towards a comprehensive asymptotic equilibrium analysis of the Boston Mechanism in random markets.

\paragraph{Robustness to Incomplete Information}
Our analysis has focused on the complete information version of the preference revelation game induced by the Boston mechanism. 
This assumption allows us to exploit the characterization linking Nash equilibrium outcomes of the Boston mechanism to the set of stable matchings under true preferences. 
Extending the analysis to incomplete information environments, where students know only their own preferences and hold beliefs about others’, remains an important direction for future research, since the Ergin–Sönmez equivalence between equilibrium outcomes and stability breaks down once information becomes private. 
In such settings, the Boston mechanism can benefit some students at the expense of others due to coordination failures in reporting behavior, and thus the overall impact of incomplete information on efficiency metrics is ambiguous.\footnote{There is important work relating the ordinal Bayesian Nash equilibria of the incomplete information game with the Nash equilibria of the complete information game, but it only applies to stable mechanisms \citep{ehlers2015matching}.}

A second question concerns whether the large efficiency losses caused by strategic behavior can be mitigated by specific equilibrium refinements. 
We note that Proposition \ref{prop:top_choice_strategic} applies to the first preference rate of any equilibrium of the induced preference revelation game; hence, under any selection, the share of students obtaining their top choice converges to zero (the only difference may lie in the speed of convergence, which is arguably of second-order importance given the dramatic decay). 
For Proposition~\ref{prop:avg_strategic}, the issue is more subtle, as the increase in average rank due to strategic behavior depends heavily on the equilibrium selected. 
We are not aware of any Nash refinement that yields exclusively the student-optimal or the school-optimal stable matchings. 
Whether such a refinement exists remains an interesting open question.

\paragraph{Robustness to Asymptotic Framework}
In this paper, we have set the number of schools to coincide with the number of students, as standard in the random market literature. An alternative large market approach would fix the number of schools $m$ and study the asymptotic properties of the Boston mechanism as the student-to-school ratio $n/m$ grows. 
In this thicker market regime, each school receives approximately $n/m$ first round applications under truthful play, so the first preference rate converges to zero even without strategic behavior.

The strategic gap between truthful and equilibrium behavior might, however, be smaller in such markets. 
When many students compete for each seat, priority orderings almost fully determine outcomes: low-priority students cannot manipulate their way into competitive schools regardless of strategy, while high-priority students face little downside from truthful reporting since desirable schools fill immediately. 
This suggests that although absolute performance deteriorates for all mechanisms in thick markets, the relative effect of strategic behavior may be less severe. 
Whether the Boston mechanism's strategic fragility is primarily a feature of balanced markets or persists in thick-market asymptotics remains an open question for future research.

	\setlength{\bibsep}{0cm}

\end{document}